\documentclass[twoside,12pt]{article}
\usepackage{indentfirst}
\usepackage{bm}
\usepackage{cite}
\usepackage{graphicx}
\usepackage{epsfig}
\usepackage{amsmath}
\usepackage{amsfonts}
\usepackage{amssymb}
\usepackage{amsthm}
\usepackage{latexsym}
\usepackage{amsmath}

\usepackage{booktabs}
\usepackage{caption}
\newtheorem{thm}{Theorem}
\newtheorem{lemma}{Lemma}
\newtheorem{corollary}{Corollary}
\usepackage{epsfig}
\usepackage{amsmath}
\usepackage{epstopdf}
\usepackage{pgf,fancyhdr}
\usepackage{float}
\begin{document}

\begin{center}
\LARGE\bf Detecting multipartite entanglement via complete orthogonal basis
\end{center}

\begin{center}
\rm  Hui Zhao,$^{*1}$ Jia Hao,$^{1}$\  Jing Li,$^{1}$ \ Shao-Ming Fei,$^{2}$ \ Naihuan Jing,$^{3}$ \ Zhi-Xi Wang$^{2}$
\end{center}

\begin{center}
\begin{footnotesize} \sl
$^1$ Department of Mathematics, Faculty of Science, Beijing University of Technology, Beijing 100124, China

$^1$ Interdisciplinary Research Institute, Faculty of Science, Beijing University of Technology, Beijing 100124, China

$^2$ School of Mathematical Sciences,  Capital Normal University,  Beijing 100048, China

$^3$ Department of Mathematics, North Carolina State University, Raleigh, NC 27695, USA

\end{footnotesize}
\end{center}
\footnotetext{\\{*}Corresponding author, zhaohui@bjut.edu.cn}
\vspace*{2mm}

\begin{center}
\begin{minipage}{15.5cm}
\parindent 20pt\footnotesize
We study genuine tripartite entanglement and multipartite entanglement in arbitrary $n$-partite quantum systems based on complete orthogonal basis (COB). While the usual Bloch representation of a density matrix uses three types of generators, the density matrix with COB operators has one uniformed type of generators which may simplify related computations. We take the advantage of this simplicity to derive useful and operational criteria to detect genuine tripartite entanglement and multipartite entanglement. We first convert the general states to simpler forms by using the relationship between general symmetric informationally complete measurements and COB. Then we derive an operational criteria to detect genuine tripartite entanglement. We study multipartite entanglement in arbitrary dimensional multipartite systems. By providing detailed examples, we demonstrate that our criteria can detect more genuine entangled and multipartite entangled states than the previously existing criteria.
\end{minipage}
\end{center}

\begin{center}
\begin{minipage}{15.5cm}
\begin{minipage}[t]{2.3cm}{\bf Keywords:}\end{minipage}
\begin{minipage}[t]{13.1cm}
Genuine tripartite entanglement, Complete orthogonal basis, Entanglement
\end{minipage}\par\vglue8pt
\end{minipage}
\end{center}

\section{Introduction}
Quantum entanglement is a fundamental phenomenon in quantum systems and plays a crucial role in various quantum information processes, including quantum cryptography \cite{ref1,ref2}, teleportation \cite{ref3} and dense coding \cite{ref4}. The genuine multipartite entanglement (GME) has remarkable properties with particular significance in quantum computation and information processing. Therefore, the measure and detection of genuine multipartite entanglement have been the essential tasks in the theory of quantum entanglement.

There have been many results in detecting entanglement and GME, such as positive partial transposition criterion \cite{ref5,ref6}, entanglement witness criterion \cite{ref7} and realignment criterion \cite{ref8}. In \cite{ref9,ref10,ref11}, the authors provided criteria for separability and k-separability in general n-partite quantum states. In \cite{ref12,ref13}, the authors derived criteria for detecting genuine tripartite entanglement based on quantum Fisher information. By using the Bloch representation of density matrices and the norms of correlation tensors, the genuine multipartite entangled criteria were presented in \cite{ref14,ref15,ref16,ref17,ref19}. Criteria for GME based on Heisenberg-Weyl representation \cite{ref20} of density matrix have been also presented. In \cite{ref21} Gour and Kalev constructed the set of general symmetric informationally complete measurements (GSICMs) from generalized Gell-Mann matrices. In \cite{ref22} the authors studied the quantum entanglement criteria by using the GSICMs for both bipartite and multipartite systems. In \cite{ref23} the authors introduced an entanglement criterion for arbitrary high-dimensional bipartite systems in terms of the GSICMs. The authors  in \cite{ref24} studied not only the arbitrary dimensional multipartite entanglement but also the arbitrary dimensional tripartite GME by using the GSICMs. In \cite{ref25} the authors demonstrated the connection between GSICMs and a complete orthogonal basis (COB).

Many approaches are based on the Bloch representation of density matrices, which becomes more complicated for high dimensional quantum systems, partly due to that the Bloch representation relies on the Gell-Mann basis given by three kinds of basis elements: the upper, diagonal and lower matrices. Different from the Bloch representations, the GSICMs consist of uniformed basis elements. Nevertheless, since the trace relationship of GSICM operators is given by certain parameters, the related calculations become complex when the dimension increases. In contrast to GSICMs, the trace relationship of COB operators depends only on the dimension. Based on the relationship between GSICMs and COB, we find that the computational complexity can be further reduced.

In this paper, we study the genuine tripartite entanglement and multipartite entanglement by using COB operators. The paper is organized as follows. In Section 2, we first review some basic concepts and present the quantum state representation in terms of COB. Then we construct matrices by using the correlation probabilities and derive the criteria of detecting the genuine tripartite entanglement. By detailed example, we show that our criteria are more efficient than the existing ones. In section 3, we present an approach to detect multipartite entanglement for arbitrary dimensional systems, which detects more multipartite entangled states than the existing criteria by detailed examples. Conclusions are given in Section 4.

\section{Genuine tripartite entanglement criterion based on COB}
We first review some basic concepts. A set of $d^2$ positive operators in a $d$-dimensional Hilbert spaces $H^{d}$, $\{P_{\alpha}\}_{\alpha=1}^{d^2} \in \mathbb{C}^{d}$, is said to be general symmetric informationally complete measurement operators if
\begin{equation}
\sum\limits_{\alpha=1}^{d^2}P_{\alpha}=\mathbb{I},
\end{equation}
\begin{equation}
Tr((P_{\alpha})^2)=a,
\end{equation}
\begin{equation}
Tr(P_{\alpha}P_{\beta})=\frac{1-da}{d(d^2-1)},
\end{equation}
where $\mathbb{I}$ is the identity operator, $\alpha,\,\beta\in\{1,2,\ldots,d^2\}$, $\alpha\neq\beta$, the parameter $a$ satisfies $\frac{1}{d^3}<a\leq\frac{1}{d^2}$,
$a=\frac{1}{d^2}$ if and only if all $P_{\alpha}$ are rank one. Given a GSICM $\{P_{\alpha}\}_{\alpha=1}^{d^2}$ and a quantum state $\rho$, one has the probabilities of measurement outcome, $p_{\alpha}=\langle P_{\alpha}\rangle=Tr(\rho P_{\alpha})$. The quantum state $\rho$ can be expressed in terms of these probabilities \cite{ref23},
\begin{equation}
\rho=\frac{d(d^{2}-1)}{ad^3-1}\sum_{\alpha=1}^{d^2}
p_{\alpha}P_{\alpha}-\frac{d(1-ad)}{ad^3-1}\mathbb{I}.
\end{equation}

An operator basis $\{A_{\alpha}\}_{\alpha=1}^{d^2}$ of Hermitian operators is called a complete orthogonal basis (COB) if
\begin{equation}\label{5}
Tr(A_{\alpha}A_{\beta})=\frac{1}{d}\delta_{\alpha\beta},
\end{equation}
\begin{equation}
\sum\limits_{\alpha=1}^{d^2}A_{\alpha}=\mathbb{I}.
\end{equation}
From $(5)$ and $(6)$, one verifies that $Tr(A_{\alpha})=\frac{1}{d}$.
It has been shown in \cite{ref25} that for any COB $\{A_{\alpha}\}_{\alpha=1}^{d^2}$, $\lambda\in(0,~\lambda^{\ast}]$ and $\lambda^{\ast}\leq\frac{1}{\sqrt{d+1}}$, the follwoing
operators
\begin{equation}\label{Lemma 1}
P_{\alpha}=\lambda A_{\alpha}+(1-\lambda)\frac{\mathbb{I}}{d^2}
\end{equation}
give rise to a GSICM.

From $(4)$ and (\ref{Lemma 1}), we have
\begin{equation}
\rho=\frac{d\lambda(d^2-1)}{ad^3-1}\sum_{\alpha=1}^{d^2} p_{\alpha}A_{\alpha}+(\frac{1}{d}-\frac{\lambda(d^2-1)}{d(ad^3-1)})\mathbb{I}.
\end{equation}
Then the probability $\mu_{\alpha}=Tr(\rho A_{\alpha})$ of obtaining the measurement outcome $\alpha$ is
\begin{equation}
\frac{\lambda(d^2-1)}{ad^3-1}p_{\alpha}+\frac{1}{d^2}-\frac{\lambda(d^2-1)}{d^2(ad^3-1)}.
\end{equation}
Therefore, a quantum state $\rho$ can be expressed as
\begin{equation}
\rho=d\sum_{\alpha=1}^{d^2}\mu_{\alpha}A_{\alpha}.
\end{equation}
Let $T^{(1)}$ be the column vector with entries $\mu_\alpha$. We have the following lemma.

\begin{lemma} For any quantum state $\rho$, the following inequality holds
\begin{equation}
\|T^{(1)}\|^{2}\leq\frac{1}{d}.
\end{equation}
\end{lemma}

\begin{proof}
For any quantum state $\rho$, we have
$Tr(\rho^{2})=Tr(\rho\rho^{\dagger})=d\|T^{(1)}\|^{2}$,
where $\dagger$ represents the conjugate transposition. Then we obtain
\begin{equation}
\|T^{(1)}\|^{2}=\frac{1}{d}Tr(\rho^{2})\leq\frac{1}{d}.
\end{equation}
When $\rho$ is a pure state, we obtain $Tr(\rho^{2})=1$ and the equality holds.
\end{proof}

For bipartite case, one can verify that the operator $A_{\alpha_{1}}^{(1)}\otimes A_{\alpha_{2}}^{(2)}$ in $H_{1}^{d_{1}}\otimes H_{2}^{d_{2}}$ is linearly independent, where $\alpha_{1}=1,2,\ldots,d_{1}^2$ and $\alpha_{2}=1,2,\ldots,d_{2}^2$. This can be seen that if $\sum_{\alpha_{1}=1}^{d_{1}^2}\sum_{\alpha_{2}=1}^{d_{2}^2}
x_{\alpha_{1}\alpha_{2}}A_{\alpha_{1}}^{(1)}\otimes A_{\alpha_{2}}^{(2)}=0$, from $(\ref{5})$ one obtains
\begin{equation}
\begin{split}
&Tr(\sum_{\alpha_{1}=1}^{d_{1}^2}\sum_{\alpha_{2}=1}^{d_{2}^2}
x_{\alpha_{1}\alpha_{2}}(A_{\alpha_{1}}^{(1)}\otimes A_{\alpha_{ 2}}^{(2)})(A_{\alpha_{1}^{'}}^{(1)}\otimes A_{\alpha_{2}^{'}}^{(2)}))\\
=&x_{11}Tr(A_{1}^{(1)}A_{\alpha_{1}^{'}}^{(1)})Tr(A_{1}^{(2)}
A_{\alpha_{2}^{'}}^{(2)})+\cdots+x_{\alpha_{1}^{'}\alpha_{2}^{'}}
Tr(A_{\alpha_{1}^{'}}^{(1)}A_{\alpha_{1}^{'}}^{(1)})
Tr(A_{\alpha_{2}^{'}}^{(2)}A_{\alpha_{2}^{'}}^{(2)})+\cdots\\
~~&+x_{d_{1}^{2}d_{2}^{2}}Tr(A_{d_{1}^{2}}^{(1)}
A_{\alpha_{1}^{'}}^{(1)})Tr(A_{d_{2}^{2}}^{(2)}A_{\alpha_{2}^{'}}^{(2)})\\
=&\frac{1}{d_{1}d_{2}}x_{\alpha_{1}^{'}\alpha_{2}^{'}}=0,
\end{split}
\end{equation}
where $\alpha_{1}^{'}=1,2,\ldots,d_{1}^2$ and $\alpha_{2}^{'}=1,2,\ldots,d_{2}^2$. Hence, $A_{\alpha_{1}}^{(1)}\otimes A_{\alpha_{2}}^{(2)}$ linearly independent and any bipartite state $\rho_{12}\in H_{1}^{d_{1}}\otimes H_{2}^{d_{2}}$ can be expressed as
\begin{equation}
\rho_{12}=d_{1}d_{2}\sum_{\alpha_{1}=1}^{d_{1}^2}
\sum_{\alpha_{2}=1}^{d_{2}^2}\mu_{\alpha_{1}\alpha_{2}}A_{\alpha_{1}}^{(1)}\otimes A_{\alpha_{2}}^{(2)},
\end{equation}
where $\mu_{\alpha_{1}\alpha_{2}}=\langle A_{\alpha_{1}}^{(1)}\otimes A_{\alpha_{2}}^{(2)}\rangle=Tr(\rho A_{\alpha_{1}}^{(1)}\otimes A_{\alpha_{2}}^{(2)})$. Let $T^{(12)}$ be the column vector with entries $\mu_{\alpha_{1}\alpha_{2}}$. We have

\begin{lemma}
For any bipartite quantum state $\rho_{12}\in H_{1}^{d_{1}}\otimes H_{2}^{d_{2}}$, the following inequality holds,
\begin{equation}
\|T^{(12)}\|^{2}\leq\frac{1}{d_{1}d_{2}}.
\end{equation}
\end{lemma}

\begin{proof}
For any bipartite quantum state $\rho_{12}$, one has
$Tr(\rho_{12}^{2})=Tr(\rho_{12}\rho_{12}^{\dagger})=d_{1}d_{2}\|T^{(12)}\|^{2}$.
Since $Tr(\rho_{12}^{2})\leq1$, we have
\begin{equation}
\|T^{(12)}\|^{2}=\frac{1}{d_{1}d_{2}}Tr(\rho_{12}^{2})\leq\frac{1}{d_{1}d_{2}},
\end{equation}
where the upper bound is attained if and only if $\rho_{12}$ is a pure state.
\end{proof}

Similarly, it can be shown that any tripartite quantum state $\rho\in H_{1}^{d_{1}}\otimes H_{2}^{d_{2}}\otimes H_{3}^{d_{3}}$ can be expressed as
\begin{equation}
\rho=d_{1}d_{2}d_{3}\sum_{\alpha_{1}=1}^{d_{1}^2}
\sum_{\alpha_{2}=1}^{d_{2}^2}\sum_{\alpha_{3}=1}^{d_{3}^2}
\mu_{\alpha_{1}\alpha_{2}\alpha_{3}}A_{\alpha_{1}}^{(1)}\otimes A_{\alpha_{2}}^{(2)}\otimes A_{\alpha_{3}}^{(3)},
\end{equation}
where $\mu_{\alpha_{1}\alpha_{2}\alpha_{3}}=Tr(\rho A_{\alpha_1}^{(1)}\otimes A_{\alpha_2}^{(2)}\otimes A_{\alpha_3}^{(3)})$. A tripartite quantum state $\rho=\sum_{z}r_{z}\rho_{f}^{z}\otimes\rho_{gh}^{z}$ is said to be biseparable under the bipartition $f|gh$, where $f\neq g\neq h\in\{1,2,3\}$, $r_{z}>0$, $\sum_{z}r_{z}=1$, $\rho_{f}^{z}$ and $\rho_{gh}^{z}$ are the density matrices in $H_{f}^{d_{f}}$ and $H_{g}^{d_{g}}\otimes H_{h}^{d_{h}}$, respectively. A quantum
state is said to be genuine tripartite entangled if it cannot be written as a convex combination of biseparable states. Let $\|\cdot\|_{tr}$ stand for the trace norm defined by  $\|A\|_{tr}=\sum\limits_{i}\xi_{i}=Tr\sqrt{AA^{\dag}}=Tr\sqrt{A^{\dag}A}$ with respect to a matrix $A\in\mathbb{R}^{m\times n}$, where $\xi_{i}$ $(i=1,2,\ldots,min\{m,n\})$ are the singular values of the matrix $A$, and $\|A\|_{tr}\leq\sqrt{\mathrm{min}\{m,n\}}\|A\|$ for any matrix $A$.

For tripartite quantum state $\rho$, $c_{11}$, $c_{12}$, $c_{21}$, $c_{22}$, $c_{31}$ and $c_{32}$ are real numbers, we define
the $d_{f}^{2}\times \{d_{f}d_{g}^{2}d_{h}^{2}\}$ matrix $B^{f|gh}$ with entries given by $\mu_{\alpha_{f}\alpha_{g}\alpha_{h}}$ and $0$, where $\alpha_{f}=1,2,\ldots,d_{f}^2$, $\alpha_{g}=1,2,\ldots,d_{g}^{2}$ and $\alpha_{h}=1,2,\ldots,d_{h}^2$. For example, for $\rho\in H_{1}^{2}\otimes H_{2}^{2}\otimes H_{3}^{2}$, we have
\begin{equation}
B^{3|12}=c_{31}
\begin{bmatrix}
\mu_{111}~~\cdots~~\mu_{441}~~0~~\cdots~~0\\
\mu_{112}~~\cdots~~\mu_{442}~~0~~\cdots~~0\\
\mu_{113}~~\cdots~~\mu_{443}~~0~~\cdots~~0\\
\mu_{114}~~\cdots~~\mu_{444}~~0~~\cdots~~0
\end{bmatrix}
+c_{32}
\begin{bmatrix}
\begin{split}
&\mu_{114}~~&\cdots~~&\mu_{444}~~&0~~~~&\cdots~~&0~~\\
&\mu_{113}~~&\cdots~~&\mu_{443}~~&0~~~~&\cdots~~&0~~\\
&~~0~~&\cdots~~&~~0~~&\mu_{112}~~&\cdots~~&\mu_{442}\\
&~~0~~&\cdots~~&~~0~~&\mu_{111}~~&\cdots~~&\mu_{441}\\
\end{split}
\end{bmatrix}.
\end{equation}

\begin{thm}
If the biseparable tripartite quantum pure state $\rho\in H_{1}^{d_1}\otimes H_{2}^{d_2}\otimes H_{3}^{d_3}$ and $f\neq g\neq h\in\{1,2,3\}$, we obtain\\
(i)if $\rho$ is separable under the bipartition $f|gh$, then
\begin{equation}
\|B^{f|gh}\|_{tr}\leq|c_{f1}|\sqrt{\frac{1}{d_{f}d_{g}d_{h}}}+|c_{f2}|\sqrt{\frac{1}{d_{g}d_{h}}},
\end{equation}
(ii)if $\rho$ is separable under the bipartition $g|fh$, then
\begin{equation}
\|B^{f|gh}\|_{tr}\leq|c_{f1}|\sqrt{\mathrm{min}\{d_{f}^{2},d_{h}^{2}\}}\sqrt{\frac{1}{d_{f}d_{g}d_{h}}}+|c_{f2}|\sqrt{\frac{d_{f}}{d_{g}d_{h}}},
\end{equation}
(iii)if $\rho$ is separable under the bipartition $h|fg$, then
\begin{equation}
\|B^{f|gh}\|_{tr}\leq|c_{f1}|\sqrt{\mathrm{min}\{d_{f}^{2},d_{g}^{2}\}}\sqrt{\frac{1}{d_{f}d_{g}d_{h}}}+|c_{f2}|\sqrt{\frac{d_{f}}{d_{g}d_{h}}}.
\end{equation}
\end{thm}

\begin{proof}
(i) If the tripartite pure state $\rho$ is separable under bipartition $f|gh$, it can be expressed as
$\rho=\rho_{f}\otimes\rho_{gh}$, where
\begin{equation}
\rho_{f}=d_{f}\sum_{\alpha_{f}=1}^{d_{f}^2}\mu_{\alpha_{f}}A_{\alpha_{f}}^{(f)},
\end{equation}
\begin{equation}
\rho_{gh}=d_{g}d_{h}\sum_{\alpha_{g}=1}^{d_{g}^2}
\sum_{\alpha_{h}=1}^{d_{h}^2}\mu_{\alpha_{g}\alpha_{h}}A_{\alpha_{g}}^{(g)}\otimes A_{\alpha_{h}}^{(h)}.
\end{equation}
Let $T^{(f)}$ and $T^{(gh)}$ be the column vectors with entries $\mu_{\alpha_{f}}$ and $\mu_{\alpha_{g}\alpha_{h}}$, respectively. Thus we have $B^{f|gh}=c_{f1}(T^{(f)})(T^{(gh)}~~0~~\cdots~~0)^{\dagger}+c_{f2}B_{2}^{gh}\otimes(T^{(gh)})^{\dagger}$, where $B_{2}^{gh}$ is a $d_{f}^{2}\times d_{f}$ matrix as follows,
\begin{equation}
B_2^{gh}=\begin{bmatrix}
\begin{split}
&~~~\mu_{d_{f}^{2}}~~&\cdots~~&\mu_{d_{f}^{2}-d_f+1}~~&\cdots~~&~~0~~&\cdots~~&~~0~~\\
&~~~~\vdots~~~&~~~&~~~~\vdots~~~&\ddots~~&~~\vdots~~&~~~&~~\vdots\\
&~~~~0~~&\cdots~~&~~~~0~~&\cdots~~&~\mu_{d_f}~~&~\cdots~~&~\mu_{1}\\
\end{split}
\end{bmatrix}^{T}.
\end{equation}
For matrix $A$, $B$ and vectors $|a\rangle$ and $|b\rangle$, we use $\|A\otimes B\|_{tr}=\|A\|_{tr}\|B\|_{tr}$,  $\|A+B\|_{tr}\leq\|A\|_{tr}+\|B\|_{tr}$ and $\||a\rangle\langle b|\|_{tr}=\||a\rangle\|\||b\rangle\|$, then we obtain
\begin{equation}
\begin{split}
\|B^{f|gh}\|_{tr}&\leq|c_{f1}|\|(T^{(f)})(T^{(gh)}~~0~~\cdots~~0)^{\dagger}\|_{tr}+|c_{f2}|\|B_{2}^{gh}\otimes(T^{(gh)})^{\dagger}\|_{tr}\\
&=|c_{f1}|\|T^{(f)}\|\|T^{(gh)}\|+|c_{f2}|\|B_{2}^{gh}\|_{tr}\|(T^{(gh)})^{\dagger}\|_{tr}\\
&\leq|c_{f1}|\|T^{(f)}\|\|T^{(gh)}\|+|c_{f2}|\sqrt{d_f}\|T^{(f)}\|\|T^{(gh)}\|\\
&\leq|c_{f1}|\sqrt{\frac{1}{d_{f}}}\sqrt{\frac{1}{d_{g}d_{h}}}+|c_{f2}|\sqrt{d_f}\sqrt{\frac{1}{d_{f}}}\sqrt{\frac{1}{d_{g}d_{h}}}\\
&=|c_{f1}|\sqrt{\frac{1}{d_{f}d_{g}d_{h}}}+|c_{f2}|\sqrt{\frac{1}{d_{g}d_{h}}},
\end{split}
\end{equation}
where we have used Lemma 1 and Lemma 2 in the third inequality.

(ii) If the tripartite pure state $\rho$ is separable under bipartition $g|fh$, it can be expressed as
$\rho=\rho_{g}\otimes\rho_{fh}$, where
\begin{equation}
\rho_{g}=d_{g}\sum_{\alpha_{g}=1}^{d_{g}^2}\mu_{\alpha_{g}}A_{\alpha_{g}}^{(g)},
\end{equation}
\begin{equation}
\rho_{fh}=d_{f}d_{h}\sum_{\alpha_{f}=1}^{d_{f}^2}
\sum_{\alpha_{h}=1}^{d_{h}^2}\mu_{\alpha_{f}\alpha_{h}}A_{\alpha_{f}}^{(f)}\otimes A_{\alpha_{h}}^{(h)}.
\end{equation}
Let $T^{(g)}$ and $T^{(fh)}$ be the column vectors with entries $\mu_{\alpha_{g}}$ and $\mu_{\alpha_{f}\alpha_{h}}$, respectively. Then we obtain $B^{f|gh}=c_{f1}(T^{(g)}~~0~~\cdots~~0)^{\dagger}\otimes B_{1}^{fh}+c_{f2}(T^{(g)})^{\dagger}\otimes (B_{21}^{fh}+B_{22}^{fh}+\cdots+B_{2d_{f}}^{fh})$, where $B_{1}^{fh}$ is a $d_{f}^{2}\times d_{h}^{2}$ matrix, $B_{21}^{fh}$, $B_{22}^{fh}$, $\ldots$, $B_{2d_{f}}^{fh}$ are $d_{f}^{2}\times \{d_{f}d_{h}^{2}\}$ matrices as follows,
\small\begin{equation}
B_1^{fh}=\begin{bmatrix}
\begin{split}
&\mu_{11}~~&\cdots~~&\mu_{1d_{h}^{2}}\\
&\mu_{21}~~&\cdots~~&\mu_{2d_{h}^{2}}\\
&~~\vdots~~&\ddots~~&~~\vdots\\
&\mu_{d_f^21}~~&\cdots~~&\mu_{d_{f}^{2}d_{h}^{2}}\\
\end{split}
\end{bmatrix},
B_{2i}^{fh}=\begin{bmatrix}
\begin{split}
&0&\cdots~~&~~~~~0&\cdots~~~~&~~~~~~0~~&\cdots~~&~~0~~\\
&\vdots&~\ddots~~&~~~~~\vdots~~~&~~~~~~&~~~~~~\vdots~~&~~~~~&~~\vdots\\
&0&\cdots~~&~~~\mu_{id_f,1}~~&\cdots~~~~&~~~\mu_{id_f,d_{h}^{2}}~~&\cdots~~&~~0\\
&\vdots&~~~~~&~~~~~\vdots~~~&~~~~~~&~~~~~~\vdots~~&~~~~~&~~\vdots\\
&0&\cdots~~&\mu_{(i-1)d_f+1,1}~~&\cdots~~~~&\mu_{(i-1)d_f+1,d_{h}^{2}}~~&\cdots~~&~~0\\
&\vdots&~~~~~&~~~~~~\vdots~~~&~~~~~~&~~~~~~\vdots~~&~\ddots~~&~~\vdots\\
&0&\cdots~~&~~~~~~0&\cdots~~~~&~~~~~~0~~&\cdots~~&~~0~~\\
\end{split}
\end{bmatrix},
\end{equation}
\small
$i=1,2,\ldots,d_f$. For matrix $A$ and $B$, we use $\|A\otimes B\|_{tr}=\|A\|_{tr}\|B\|_{tr}$, $\|A+B\|_{tr}\leq\|A\|_{tr}+\|B\|_{tr}$ and $\|A_{m\times n}\|_{tr}\leq \sqrt{\mathrm{min}\{m,n\}}\|A\|$, thus
\begin{equation}
\begin{split}
\|B^{f|gh}\|_{tr}&\leq|c_{f1}|\|(T^{(g)}~~0~~\cdots~~0)^{\dagger}\otimes B^{fh}\|_{tr}+|c_{f2}|\|(T^{(g)})^{\dagger}\otimes (B_{21}^{fh}+B_{22}^{fh}+\cdots+B_{2d_{f}}^{fh})\|_{tr}\\
&=|c_{f1}|\|(T^{(g)})^{\dagger}\|_{tr}\|B^{fh}\|_{tr}+|c_{f2}|\|(T^{(g)})^{\dagger}\|_{tr}\|(B_{21}^{fh}+B_{22}^{fh}+\cdots+B_{2d_{f}}^{fh})\|_{tr}\\
&\leq|c_{f1}|\sqrt{\mathrm{min}\{d_{f}^{2},d_{h}^{2}\}}\|T^{(g)}\|\|T^{(fh)}\|+|c_{f2}|\|T^{(g)}\|d_{f}\|T^{(fh)}\|\\
&\leq|c_{f1}|\sqrt{\mathrm{min}\{d_{f}^{2},d_{h}^{2}\}}\sqrt{\frac{1}{d_{g}}}\sqrt{\frac{1}{d_{f}d_{h}}}+|c_{f2}|\sqrt{\frac{1}{d_{g}}}d_{f}\sqrt{\frac{1}{d_{f}d_{h}}}\\
&=|c_{f1}|\sqrt{\mathrm{min}\{d_{f}^{2},d_{h}^{2}\}}\sqrt{\frac{1}{d_{f}d_{g}d_{h}}}+|c_{f2}|\sqrt{\frac{d_{f}}{d_{g}d_{h}}}.
\end{split}
\end{equation}

(iii) Using similar method, if $\rho$ is separable under the bipartition $h|fg$, we obtain $B^{f|gh}=c_{f1}[B_1^{fg}~~0_{d_{f}^{2}\times\{(d_f-1)d_{g}^{2}\}}]\otimes(T^{(h)})^{\dagger}+c_{f2}(B_{21}^{fg}+B_{22}^{fg}+\cdots+B_{2d_{f}}^{fg})\otimes(T^{(h)})^{\dagger}$, where $B_1^{fg}$ is a $d_{f}^{2}\times d_{g}^{2}$ matrix, $B_{21}^{fg}$, $B_{22}^{fg}$, $\ldots$, $B_{2d_{f}}^{fg}$ are $d_{f}^{2}\times \{d_fd_{g}^{2}\}$ matrix. Thus
\begin{equation}
\|B^{f|gh}\|_{tr}\leq|c_{f1}|\sqrt{\mathrm{min}\{d_{f}^{2},d_{g}^{2}\}}\sqrt{\frac{1}{d_{f}d_{g}d_{h}}}+|c_{f2}|\sqrt{\frac{d_{f}}{d_{g}d_{h}}}.
\end{equation}
\end{proof}

We are now ready to consider genuine tripartite entanglement. Set $Q_1=\mathrm{max}\{|c_{11}|\sqrt{\frac{1}{d_{1}d_{2}d_{3}}}+|c_{12}|\sqrt{\frac{1}{d_{2}d_{3}}},|c_{11}|\sqrt{\mathrm{min}\{d_{1}^{2},d_{3}^{2}\}}\sqrt{\frac{1}{d_{1}d_{2}d_{3}}}+|c_{12}|\sqrt{\frac{d_1}{d_{2}d_{3}}},|c_{11}|\sqrt{\mathrm{min}\{d_{1}^{2},d_{2}^{2}\}}\sqrt{\frac{1}{d_{1}d_{2}d_{3}}}++|c_{12}|\sqrt{\frac{d_1}{d_{2}d_{3}}}\}$, $Q_2=\mathrm{max}\{|c_{21}|\sqrt{\frac{1}{d_{1}d_{2}d_{3}}}+|c_{22}|\sqrt{\frac{1}{d_{1}d_{3}}},|c_{21}|\sqrt{\mathrm{min}\{d_{2}^{2},d_{3}^{2}\}}\sqrt{\frac{1}{d_{1}d_{2}d_{3}}}+|c_{22}|\sqrt{\frac{d_{2}}{d_{1}d_{3}}},|c_{21}|\sqrt{\mathrm{min}\{d_{2}^{2},d_{1}^{2}\}}\\\sqrt{\frac{1}{d_{1}d_{2}d_{3}}}+|c_{22}|\sqrt{\frac{d_{2}}{d_{1}d_{3}}}\}$, $Q_3=\mathrm{max}\{|c_{31}|\sqrt{\frac{1}{d_{1}d_{2}d_{3}}}+|c_{32}|\sqrt{\frac{1}{d_{1}d_{2}}},|c_{31}|\sqrt{\mathrm{min}\{d_{3}^{2},d_{2}^{2}\}}\sqrt{\frac{1}{d_{1}d_{2}d_{3}}}+|c_{32}|\\\sqrt{\frac{d_{3}}{d_{1}d_{2}}},|c_{31}|\sqrt{\mathrm{min}\{d_{3}^{2},d_{1}^{2}\}}\sqrt{\frac{1}{d_{1}d_{2}d_{3}}}+|c_{32}|\sqrt{\frac{d_{3}}{d_{1}d_{2}}}\}$ and $B(\rho)=\frac{1}{3}(\|B^{1|23}\|_{tr}+\|B^{2|13}\|_{tr}+\|B^{3|12}\|_{tr})$, where $c_{11}$, $c_{12}$, $c_{21}$, $c_{22}$, $c_{31}$ and $c_{32}$ are real numbers. We have the following Theorem.

\begin{thm}
A quantum mixed state $\rho\in H_{1}^{d_{1}}\otimes H_{2}^{d_{2}}\otimes H_{3}^{d_{3}}$ is genuine tripartite entangled if $B(\rho)>\frac{1}{3}(Q_1+Q_2+Q_3)$.
\end{thm}

\begin{proof}
If $\rho$ is a biseparable, one has
$\rho=\sum_{i}o_{i}\rho_{i}^{1}\otimes\rho_{i}^{23}
+\sum_{j}r_{j}\rho_{j}^{2}\otimes\rho_{j}^{13}+\sum_{k}s_{k}\rho_{k}^{3}\otimes\rho_{k}^{12}$ with $0\leq o_{i},r_{j},s_{k}\leq1$ and $\sum_{i}o_{i}+\sum_{j}r_{j}+\sum_{k}s_{k}=1$. By Theorem 1, we have that
\begin{equation}
\begin{split}
B(\rho)=&\frac{1}{3}\left(\|B^{1|23}(\rho)\|_{tr}+\|B^{2|13}(\rho)\|_{tr}+\|B^{3|12}(\rho)\|_{tr}\right)\\
=&\frac{1}{3}\left[\|B^{1|23}\left(\sum_{i}o_{i}\rho_{i}^{1}\otimes\rho_{i}^{23}+\sum_{j}r_{j}\rho_{j}^{2}\otimes\rho_{j}^{13}+\sum_{k}s_{k}\rho_{k}^{3}\otimes\rho_{k}^{12}\right)\|_{tr}\right.\\
&\left.+\|B^{2|13}\left(\sum_{i}o_{i}\rho_{i}^{1}\otimes\rho_{i}^{23}+\sum_{j}r_{j}\rho_{j}^{2}\otimes\rho_{j}^{13}+\sum_{k}s_{k}\rho_{k}^{3}\otimes\rho_{k}^{12}\right)\|_{tr}\right.\\
&\left.+\|B^{3|12}\left(\sum_{i}o_{i}\rho_{i}^{1}\otimes\rho_{i}^{23}+\sum_{j}r_{j}\rho_{j}^{2}\otimes\rho_{j}^{13}+\sum_{k}s_{k}\rho_{k}^{3}\otimes\rho_{k}^{12}\right)\|_{tr}\right]\\
\leq&\frac{1}{3}\left[\sum_{i}o_{i}\|B^{1|23}(\rho_{i}^{1}\otimes\rho_{i}^{23})\|_{tr}+\sum_{j}r_{j}\|B^{1|23}(\rho_{j}^{2}\otimes\rho_{j}^{13})\|_{tr}+\sum_{k}s_{k}\|B^{1|23}(\rho_{k}^{3}\otimes\rho_{k}^{12})\|_{tr}\right.\\
&\left.+\sum_{i}o_{i}\|B^{2|13}(\rho_{i}^{1}\otimes\rho_{i}^{23})\|_{tr}+\sum_{j}r_{j}\|B^{2|13}(\rho_{j}^{2}\otimes\rho_{j}^{13})\|_{tr}+\sum_{k}s_{k}\|B^{2|13}(\rho_{k}^{3}\otimes\rho_{k}^{12})\|_{tr}\right.\\
&\left.+\sum_{i}o_{i}\|B^{3|12}(\rho_{i}^{1}\otimes\rho_{i}^{23})\|_{tr}+\sum_{j}r_{j}\|B^{3|12}(\rho_{j}^{2}\otimes\rho_{j}^{13})\|_{tr}+\sum_{k}s_{k}\|B^{3|12}(\rho_{k}^{3}\otimes\rho_{k}^{12})\|_{tr}\right]\\
\leq&\frac{1}{3}\left[\left(\sum_{i}o_{i}+\sum_{j}r_{j}+\sum_{k}s_{k}\right)Q_1+\left(\sum_{i}o_{i}+\sum_{j}r_{j}+\sum_{k}s_{k}\right)Q_2\right.\\
&\left.+\left(\sum_{i}o_{i}+\sum_{j}r_{j}+\sum_{k}s_{k}\right)Q_3\right]=\frac{1}{3}(Q_1+Q_2+Q_3).\\
\end{split}
\end{equation}

Consequently, if $B(\rho)>\frac{1}{3}(Q_1+Q_2+Q_3)$, $\rho$ is genuine tripartite entangled.
\end{proof}

In particular, for states that are invariant under any permutation of basis, $\rho=\rho^p=p \rho p^{\dagger}$, where $p$ is any permutation operator. Let $d_1=d_2=d_3=d$, $c_{11}=c_{21}=c_{31}$ and $c_{12}=c_{22}=c_{32}$, the following corollary can be deduced.

\begin{corollary}
If a permutational invariant mixed state is biseparable, then we have
\begin{equation}B(\rho)=\frac{1}{3}(\|B^{1|23}\|_{tr}+\|B^{2|13}\|_{tr}+\|B^{3|12}\|_{tr})\leq \frac{1}{3}(c_{11}\sqrt{\frac{1}{d^3}}+c_{12}\frac{1}{d}+2c_{11}\sqrt{\frac{1}{d}}+2c_{12}\sqrt{\frac{1}{d}}).
\end{equation}
Therefore if $B(\rho)>\frac{1}{3}(c_{11}\sqrt{\frac{1}{d^3}}+c_{12}\frac{1}{d}+2c_{11}\sqrt{\frac{1}{d}}+2c_{12}\sqrt{\frac{1}{d}})$, $\rho$ is genuine tripartite entangled.

\end{corollary}

To illustrate our criterion, we consider the case of $d=2$. We use the four matrices provided in \cite{ref25} through construction 1,
\begin{equation}\label{33}
\begin{split}
& A_{1}=
\begin{bmatrix}
\frac{1}{2} & \frac{1-\mathbf{i}}{4} \\
\frac{1+\mathbf{i}}{4} & 0 \\
\end{bmatrix},~~~~
A_{2}=
\begin{bmatrix}
0 & \frac{-1-\mathbf{i}}{4} \\
\frac{-1+\mathbf{i}}{4} & \frac{1}{2} \\
\end{bmatrix},\\
& A_{3}=
\begin{bmatrix}
0 & \frac{1+\mathbf{i}}{4} \\
\frac{1-\mathbf{i}}{4} & \frac{1}{2} \\
\end{bmatrix},~~~~
A_{4}=
\begin{bmatrix}
\frac{1}{2} & \frac{-1+\mathbf{i}}{4} \\
\frac{-1-\mathbf{i}}{4} & 0 \\
\end{bmatrix},
\end{split}
\end{equation}
where $\mathbf{i}=\sqrt{-1}$.

\textbf{Example 1} Consider the $2\times2\times2$ quantum state $\rho_{GHZ}$,
\begin{equation}\label{30}
\rho_{GHZ}=\frac{x}{8}\mathbb{I}_{8}+(1-x)|GHZ\rangle\langle GHZ|,~~0\leq x\leq1,
\end{equation}
where $|GHZ\rangle=\frac{1}{\sqrt{2}}[|000\rangle+|111\rangle]$. Let $c_{11}=c_{21}=c_{31}=1$ and $c_{12}=c_{22}+c_{32}=0$, by Corollary 1 we have $f_{1}(x)=B(\rho)-\frac{1}{3}(2\sqrt{\frac{1}{2}}+\sqrt{\frac{1}{8}})=\frac{\sqrt{x^{2}-2x+2}}{8}
+\frac{3\sqrt{2}|x-1|}{8}-\frac{1}{3}(2\sqrt{\frac{1}{2}}+\sqrt{\frac{1}{8}})>0$. Thus $\rho$ is genuine tripartite entangled for $0\leq x<0.1919$. The Theorem 1 in \cite{ref16} implies that $\rho$ is genuine tripartite entangled if $g_{1}(x)=2-2x-\sqrt{3}>0$, i.e., $0\leq x<0.134$. Theorem 1 in \cite{ref26} shows that $\rho$ is genuine tripartite entangled if $g_{2}(x)=C_3(\rho)-\frac{1}{2}\sqrt{6-25x+\frac{25}{2}x^2}\leq0$, i.e., $0\leq x<0.08349$. Our result clearly outperforms these two results, see Figure 1.
\begin{figure}[h]
  \centering
  \includegraphics[width=14cm]{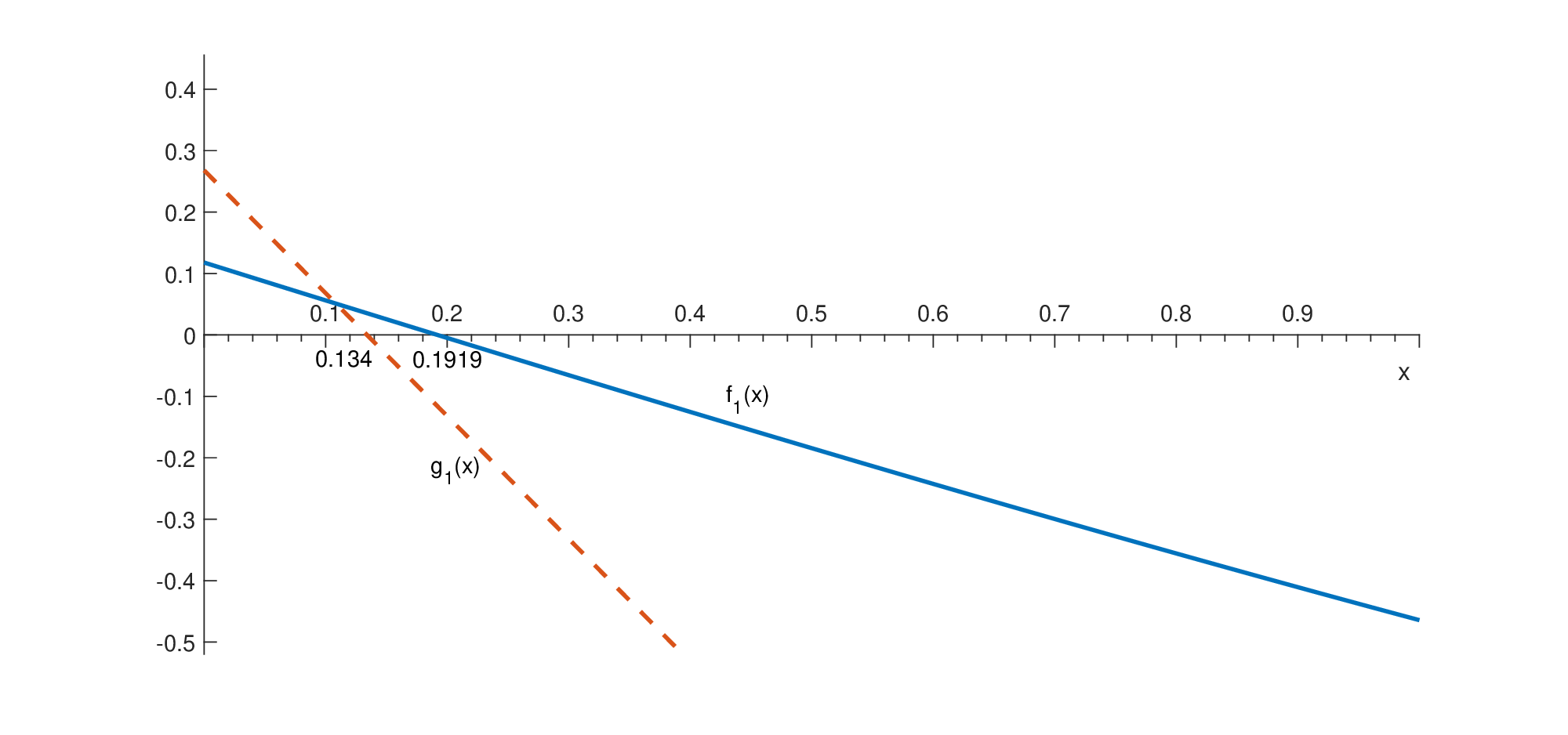}
  \caption{$f_{1}(x)$ from our result (solid blue line) and $g_{1}(x)$ in \cite{ref16} (dotted red line).}\label{1}
\end{figure}

The range of genuine entanglement for (\ref{30}) given in \cite{ref29} is better than our conclusion. However, \cite{ref29} only considered the qubit case. Our method can detect genuine tripartite entanglement for arbitrary dimensions. We take the following $3\times 3\times 2$ quantum state to illustrate our method. We use the matrices provided in \cite{ref25} through construction 2. When $d=3$,
\begin{equation}\footnotesize
\begin{split}
& A_1=\left[\begin{array}{ccc}
-\frac{2}{9} & \frac{-7+3 \mathbf{i} \sqrt{7}}{84 \sqrt{3}} & \frac{\mathbf{i}}{6 \sqrt{5}}-\frac{1}{6 \sqrt{7}} \\
\frac{-7-3 \mathbf{i}\sqrt{7}}{84 \sqrt{3}} & \frac{1}{9} & -\frac{-5 \mathbf{i}+\sqrt{15}}{30 \sqrt{2}} \\
-\frac{\mathbf{i}}{6 \sqrt{5}}-\frac{1}{6 \sqrt{7}} & -\frac{5 \mathbf{i}+\sqrt{15}}{30 \sqrt{2}} & \frac{4}{9}
\end{array}\right], \quad A_2=\left[\begin{array}{cccc}
\frac{1}{9} & \frac{2}{3 \sqrt{3}} & 0 \\
\frac{2}{3 \sqrt{3}} & \frac{1}{9} & 0 \\
0 & 0 & \frac{1}{9}
\end{array}\right], \\
& A_3=\left[\begin{array}{ccc}
\frac{1}{9} & \frac{-1-3\mathbf{i} \sqrt{7}}{12 \sqrt{3}} & 0 \\
\frac{-1+3 \mathbf{i} \sqrt{7}}{12 \sqrt{3}} & \frac{1}{9} & 0 \\
0 & 0 & \frac{1}{9}
\end{array}\right], \quad A_4=\left[\begin{array}{ccc}
\frac{1}{9} & \frac{-7+3 \mathbf{i} \sqrt{7}}{84 \sqrt{3}} & \frac{1}{\sqrt{7}} \\
\frac{-7-3 \mathbf{i} \sqrt{7}}{84 \sqrt{3}} & \frac{1}{9} & 0 \\
\frac{1}{\sqrt{7}} & 0 & \frac{1}{9}
\end{array}\right], \\
&A_5=\left[\begin{array}{ccc}
\frac{1}{9} & \frac{-7+3 \mathbf{i} \sqrt{7}}{84 \sqrt{3}} & -\frac{\mathbf{i}\sqrt{5}}{6}-\frac{1}{6 \sqrt{7}} \\
\frac{-7-3 \mathbf{i}\sqrt{7}}{84 \sqrt{3}} & \frac{1}{9} & 0 \\
\frac{\mathbf{i} \sqrt{5}}{6}-\frac{1}{6 \sqrt{7}} & 0 & \frac{1}{9}
\end{array}\right] , \quad  A_6=\left[\begin{array}{ccc}
\frac{1}{9} & \frac{-7+3 \mathbf{i} \sqrt{7}}{84 \sqrt{3}} & \frac{\mathbf{i}}{6 \sqrt{5}}-\frac{1}{6 \sqrt{7}} \\
\frac{-7-3 \mathbf{i}\sqrt{7}}{84 \sqrt{3}} & \frac{1}{9} & \sqrt{\frac{2}{15}} \\
-\frac{\mathbf{i}}{6 \sqrt{5}}-\frac{1}{6 \sqrt{7}} & \sqrt{\frac{2}{15}} & \frac{1}{9}
\end{array}\right]\\
&A_7=\left[\begin{array}{ccc}
\frac{1}{9} & \frac{-7+3 \mathbf{i} \sqrt{7}}{84 \sqrt{3}} & \frac{\mathbf{i}}{6 \sqrt{5}}-\frac{1}{6 \sqrt{7}} \\
\frac{-7-3 \mathbf{i}\sqrt{7}}{84 \sqrt{3}} & \frac{1}{9} & -\frac{15 \mathbf{i}+\sqrt{15}}{30 \sqrt{2}} \\
-\frac{\mathbf{i}}{6 \sqrt{5}}-\frac{1}{6 \sqrt{7}} & -\frac{-15 \mathbf{i}+\sqrt{15}}{30 \sqrt{2}} & \frac{1}{9}
\end{array}\right], \quad  A_8=\left[\begin{array}{ccc}
\frac{4}{9} & \frac{-7+3\mathbf{i} \sqrt{7}}{84 \sqrt{3}} & \frac{\mathbf{i}}{6 \sqrt{5}}-\frac{1}{6 \sqrt{7}} \\
\frac{-7-3\mathbf{i}\sqrt{7}}{84 \sqrt{3}} & -\frac{2}{9} & -\frac{-5 \mathbf{i}+\sqrt{15}}{30 \sqrt{2}} \\
-\frac{\mathbf{i}}{6 \sqrt{5}}-\frac{1}{6 \sqrt{7}} & -\frac{5 \mathbf{i}+\sqrt{15}}{30 \sqrt{2}} & \frac{1}{9}
\end{array}\right],\\
& A_9=\left[\begin{array}{ccc}
\frac{1}{9} & \frac{-7+3 \mathbf{i}\sqrt{7}}{84 \sqrt{3}} & \frac{\mathbf{i}}{6 \sqrt{5}}-\frac{1}{6 \sqrt{7}} \\
\frac{-7-3 \mathbf{i} \sqrt{7}}{84 \sqrt{3}} & \frac{4}{9} & -\frac{-5\mathbf{i}+\sqrt{15}}{30 \sqrt{2}} \\
-\frac{\mathbf{i}}{6 \sqrt{5}}-\frac{1}{6 \sqrt{7}} & -\frac{5 \mathbf{i}+\sqrt{15}}{30 \sqrt{2}} & -\frac{2}{9}
\end{array}\right].\\
&
\end{split}
\end{equation}
When $d=2$,
\begin{equation}
\begin{aligned}
& A_1=\left[\begin{array}{cc}
\frac{1}{4}-\frac{1}{2 \sqrt{2}}{ } & -\frac{1}{4 \sqrt{3}}+\frac{\mathbf{i}}{2 \sqrt{6}} \\
-\frac{1}{4 \sqrt{3}}-\frac{\mathbf{i}}{2 \sqrt{6}} & \frac{1}{4}+\frac{1}{2 \sqrt{2}}
\end{array}\right],  A_2=\left[\begin{array}{cc}
\frac{1}{4} & \frac{\sqrt{3}}{4} \\
\frac{\sqrt{3}}{4} & \frac{1}{4}
\end{array}\right], \\
& A_3=\left[\begin{array}{cc}
\frac{1}{4} & -\frac{1}{4 \sqrt{3}}-\frac{\mathbf{i}}{\sqrt{6}} \\
-\frac{1}{4 \sqrt{3}}+\frac{\mathbf{i}}{\sqrt{6}} & \frac{1}{4}
\end{array}\right],
 A_4=\left[\begin{array}{cc}
\frac{1}{4}+\frac{1}{2 \sqrt{2}} & -\frac{1}{4 \sqrt{3}}+\frac{\mathbf{i}}{2 \sqrt{6}} \\
-\frac{1}{4 \sqrt{3}}-\frac{\mathbf{i}}{2 \sqrt{6}} & \frac{1}{4}-\frac{1}{2 \sqrt{2}}
\end{array}\right] .
\end{aligned}
\end{equation}

\textbf{Example 2} Consider the $3\times 3\times 2$ quantum state,
\begin{equation}
\rho=\frac{1-x}{18} I+x|\varphi\rangle\langle\varphi|,
\end{equation}
where $|\varphi\rangle=\frac{1}{\sqrt{5}}[(|10\rangle+|21\rangle)|0\rangle
+(|00\rangle+|11\rangle+|22\rangle)|1\rangle]$.

From Theorem 1(i), when $c_{31}=0$ and $c_{32}=1$, we have $\|B^{3|12}\|_{tr}>\sqrt{\frac{1}{9}}$. We can detect the entanglement of $\rho$ for $0.496<x\leq1$. In Ref. \cite{ref30} the entanglement is only detected for $x=1$. Hence, our method can detect more entanglement. Using Theorem 2 and setting $c_{11}=c_{12}=c_{21}=c_{22}=c_{31}=0$ and $c_{32}=1$, we get that $f(x)=B(\rho)-\sqrt{\frac{2}{9}}$. $\rho$ is genuine entangled for $f(x) > 0$, i.e., $0.7152 < x\leq1$. While the criterion given in \cite{ref30} can not detect the genuine tripartite entanglement in this case.

\section{Multipartite entanglement criterion based on COB}
Now we consider multipartite entanglement in $n$-partite quantum systems. Let $\{A_{\alpha_{s}}^{(s)}\}_{\alpha_{s}=1}^{d_{s}^{2}}$ be the COB of the $s$th $d_{s}$-dimensional Hilbert space $H_{s}^{d_{s}}$. It is direct to verify that the set of operators $A_{\alpha_{1}}^{(1)}\otimes A_{\alpha_{2}}^{(2)}\otimes\cdots\otimes A_{\alpha_{n}}^{(n)}$ are linearly independent. Therefore, any $n$-partite quantum state $\rho\in H_{1}^{d_{1}}\otimes H_{2}^{d_{2}}\otimes \cdots\otimes H_{n}^{d_{n}}$ can be expressed as
\begin{equation}
\rho=d_{1}d_{2}\cdots d_{n}\sum\limits_{\alpha_{1}=1}^{d_{1}^{2}}
\sum\limits_{\alpha_{2}=1}^{d_{2}^{2}}\cdots\sum\limits_{\alpha_{n}=1}^{d_{n}^{2}}
\mu_{\alpha_{1}\alpha_{2}\cdots\alpha_{n}}A_{\alpha_{1}}^{(1)}\otimes A_{\alpha_{2}}^{(2)}\otimes\cdots\otimes A_{\alpha_{n}}^{(n)},
\end{equation}
where $\mu_{\alpha_{1}\alpha_{2}\cdots\alpha_{n}}=Tr(\rho A_{\alpha_{1}}^{(1)}\otimes A_{\alpha_{2}}^{(2)}\otimes\cdots\otimes A_{\alpha_{n}}^{(n)})$. Let $T^{(12\cdots n)}$ be the column vector with entries $\mu_{\alpha_{1}\alpha_{2}\cdots\alpha_{n}}$. An $n$-partite quantum state $\rho=\sum_{z}r_{z}\rho_{l_{1}\cdots l_{v_{1}}}^{z}\otimes\rho_{l_{v_{1}+1}\cdots l_{v_{1}+v_{2}}}^{z}\otimes\cdots\otimes\rho_{l_{n-v_{k}+1}\cdots l_{n}}^{z}$ is said to be $k$-separable under the partition $l_{1}\cdots l_{v_{1}}|l_{v_{1}+1}\cdots l_{v_{1}+v_{2}}|\cdots|l_{n-v_{k}+1}\cdots l_{n}$, where $\sum_{i=1}^{k}v_{i}=n$, $\rho_{l_{1}\cdots l_{v_{1}}}^{z}$, $\rho_{l_{v_{1}+1}\cdots l_{v_{1}+v_{2}}}^{z}$, $\ldots$, and $\rho_{l_{n-v_{k}+1}\cdots l_{n}}^{z}$ are pure states in $H_{l_{1}}^{d_{l_{1}}}\otimes\cdots\otimes H_{l_{v_{1}}}^{d_{l_{v_{1}}}}$, $H_{l_{v_{1}+1}}^{d_{l_{v_{1}+1}}}\otimes\cdots\otimes H_{l_{v_{1}+v_{2}}}^{d_{l_{v_{1}+v_{2}}}}$, $\ldots$, and $H_{l_{n-v_{k}+1}}^{d_{l_{n-v_{k}+1}}}\otimes\cdots\otimes H_{l_{n}}^{d_{l_{n}}}$, respectively, $l_{1}\neq l_{2}\neq\cdots\neq l_{n}\in\{1,2,\ldots,n\}$.
\begin{lemma}
For any $n$-partite quantum state $\rho$, the following inequality hold,
\begin{equation}
\|T^{(12\cdots n)}\|^{2}\leq\frac{1}{d_{1}d_{2}\cdots d_{n}}.
\end{equation}
\end{lemma}
\begin{proof}
For any $n$-partite quantum state $\rho$, we have
$Tr(\rho^{2})=Tr(\rho\rho^{\dagger})=d_{1}d_{2}\cdots d_{n}\|T^{(12\cdots n)}\|^{2}$.
Since $Tr(\rho^{2})\leq1$,
\begin{equation}\label{29}
\|T^{(12\cdots n)}\|^{2}=\frac{1}{d_{1}d_{2}\cdots d_{n}}Tr(\rho^{2})\leq\frac{1}{d_{1}d_{2}\cdots d_{n}}.
\end{equation}
The upper bound is attained if and only if $\rho$ is a pure state.
\end{proof}

If an $n$-partite state $\rho\in H_{1}^{d_{1}}\otimes H_{2}^{d_{2}}\otimes \cdots\otimes H_{n}^{d_{n}}$ is fully separable. We define $B^{1|2|\cdots|n}$ is a $d_1^2\times \{d_2^2d_3^2\cdots d_n^2\}$ matrix with entries given by
$(B^{1|2|\cdots|n})_{\alpha_{1}\alpha_{2}\cdots\alpha_{n}}=\mu_{\alpha_{1}\alpha_{2}\cdots\alpha_{n}}$. For example, for $\rho\in H_{1}^{2}\otimes H_{2}^{2}\otimes H_{3}^{2}\otimes H_{4}^{2}$, we obtain
\begin{equation}
B^{1|2|3|4}=
\begin{bmatrix}
\mu_{1111}~~\mu_{1112}~~\cdots~~\mu_{1444}\\
\mu_{2111}~~\mu_{2112}~~\cdots~~\mu_{2444}\\
\mu_{3111}~~\mu_{3112}~~\cdots~~\mu_{3444}\\
\mu_{4111}~~\mu_{4112}~~\cdots~~\mu_{4444}
\end{bmatrix}.
\end{equation}

\begin{thm}
If the $n$-partite quantum state $\rho$ is fully separable, then we get
\begin{equation}
\|B^{1|2|\cdots|n}\|_{tr}\leq\sqrt{\frac{1}{d_1d_2\cdots d_n}}.
\end{equation}
\end{thm}
\begin{proof}
If the $n$-partite state $\rho$ is fully separable, it can be expressed as
$\sum\limits_{z}r_{z}\rho_{1}^{z}\otimes\rho_{2}^{z}\otimes\cdots\otimes\rho_{n}^{z}$, $0<r_{z}\leq1$, $\sum\limits_{z}r_{z}=1$, where
\begin{equation}
\rho_{l}^{z}=d_{l}\sum\limits_{\alpha_{l=1}}^{d_{l}^{2}}
\mu_{\alpha_{l}}^{z}A_{\alpha_{l}}^{(l)}, ~~~l\in\{1,2,\ldots,n\}.
\end{equation}
Let $T^{(l)}$ be the column vectors with entries of $\mu_{\alpha_{l}}^{z}$. Then we obtain
\begin{equation}
\begin{split}
\|B^{1|2|\cdots|n}\|_{tr}&\leq\sum\limits_{z}r_{z}\|T^{(1)}\otimes(T^{(2)})^{\dag}\otimes\cdots\otimes(T^{(n)})^{\dag}\|_{tr}=\sum\limits_{z}r_{z}\|T^{(1)}\|_{tr}\|(T^{(2)})^{\dag}\|_{tr}\cdots\|(T^{(n)})^{\dag}\|_{tr}\\
&=\sum\limits_{z}r_{z}\|T^{(1)}\|\|T^{(2)}\|\cdots\|T^{(n)}\|\leq\sum\limits_{z}r_{z}\sqrt{\frac{1}{d_1}}\sqrt{\frac{1}{d_2}}\cdots\sqrt{\frac{1}{d_n}}=\sqrt{\frac{1}{d_1d_2\cdots d_n}},
\end{split}
\end{equation}
where we have used $\|A\otimes B\|_{tr}=\|A\|_{tr}\|B\|_{tr}$ for matrices $A$ and $B$ and Lemma 1 in the second inequality.
\end{proof}

If an $n$-partite state $\rho\in H_{1}^{d_{1}}\otimes H_{2}^{d_{2}}\otimes \cdots\otimes H_{n}^{d_{n}}$ is biseparable under the bipartition $l_{1}|l_{2}\cdots l_{n}$, we define
a $d_{l_1}^{2}\times \{d_{l_2}^{2}\cdots d_{l_n}^{2}\}$ matrix $B^{l_{1}|l_{2}\cdots l_{n}}$ with entries given by
$(B^{l_{1}|l_{2}\cdots l_{n}})_{\alpha_{1}\alpha_{2}\cdots\alpha_{n}}
=\mu_{\alpha_{1}\alpha_{2}\cdots\alpha_{n}}$. For example, for $\rho\in H_{1}^{2}\otimes H_{2}^{2}\otimes H_{3}^{2}\otimes H_{4}^{2}$, we have
\begin{equation}
B^{2|134}=
\begin{bmatrix}
\mu_{1111}~~\mu_{1112}~~\cdots~~\mu_{4144}\\
\mu_{1211}~~\mu_{1212}~~\cdots~~\mu_{4244}\\
\mu_{1311}~~\mu_{1312}~~\cdots~~\mu_{4344}\\
\mu_{1411}~~\mu_{1412}~~\cdots~~\mu_{4444}
\end{bmatrix}.
\end{equation}

If $\rho$ is $k$-separable under the bipartition $l_{1}\cdots l_{v_{1}}|l_{v_{1}+1}\cdots l_{v_{1}+v_{2}}|\cdots|l_{n-v_{k}+1}\cdots l_{n}$, we denote $B^{l_{1}\cdots l_{v_{1}}|l_{v_{1}+1}\cdots l_{v_{1}+v_{2}}|\cdots|l_{n-v_{k}+1}\cdots l_{n}}=(B^{l_{1}\cdots l_{v_{1}}|l_{v_{1}+1}\cdots l_{v_{1}+v_{2}}|\cdots|l_{n-v_{k}+1}\cdots l_{n}})_{\alpha_{l_{1}}\alpha_{l_{2}}\cdots\alpha_{l_{n}}}=\mu_{\alpha_{l_{1}}\alpha_{l_{2}}\cdots\alpha_{l_{n}}}$ is a $\{d_{l_{n-v_{k}+1}}^{2}\cdots d_{l_{n-1}}^{2}d_{l_{n}}\}\times\{d_{l_{1}}^{2}d_{l_{2}}^{2}\cdots d_{l_{n-1}}^{2}\}$ matrix. For example, for $\rho\in H_{1}^{2}\otimes H_{2}^{2}\otimes H_{3}^{2}\otimes H_{4}^{2}$, we obtain
\begin{equation}
B^{12|34}=
\begin{bmatrix}
&\mu_{1111}~~&\mu_{1113}~~\cdots~~&\mu_{4411}~~&\mu_{4413}\\
&\mu_{1121}~~&\mu_{1123}~~\cdots~~&\mu_{4421}~~&\mu_{4423}\\
&\vdots~~&\vdots~~~~~~~~&\vdots~&\vdots\\
&\mu_{1142}~~&\mu_{1144}~~\cdots~~&\mu_{4442}~~&\mu_{4444}\\
\end{bmatrix}.
\end{equation}

\begin{thm} If an $n$-partite quantum state is $k$-separable under the bipartition $l_{1}\cdots l_{v_{1}}|l_{v_{1}+1}\cdots\\ l_{v_{1}+v_{2}}|\cdots|l_{n-v_{k}+1}\cdots l_{n}$, then the following inequalities hold,\\
(i) $\|B^{l_{1}|l_{2}\cdots l_{n}}\|_{tr}\leq \sqrt{\frac{1}{d_{l_{1}}d_{l_{2}}\cdots d_{l_{n}}}}$,\\
(ii) $\|B^{l_{1}\cdots l_{v_{1}}|l_{v_{1}+1}\cdots l_{v_{1}+v_{2}}|\cdots|l_{n-v_{k}+1}\cdots l_{n}}\|_{tr}\leq \sqrt{\frac{1}{d_{l_{1}}d_{l_{2}}\cdots d_{l_{n-1}}}}.$\\
\end{thm}

\begin{proof}
(i) If the $n$-partite state $\rho$ is biseparable under the bipartition $l_{1}|l_{2}\cdots l_{n}$, it can be expressed as
$\sum\limits_{z}r_{z}\rho_{l_{1}}^{z}\otimes\rho_{l_{2}\cdots l_{n}}^{z}$, $0<r_{z}\leq1$, $\sum\limits_{z}r_{z}=1$, where
\begin{equation}
\rho_{l_{1}}^{z}=d_{l_{1}}\sum\limits_{\alpha_{l_{1}=1}}^{d_{l_{1}}^{2}}
\mu_{\alpha_{l_{1}}}^{z}A_{\alpha_{l_{1}}}^{(l_{1})},
\end{equation}
\begin{equation}
\rho_{l_{2}\cdots l_{n}}^{z}=d_{l_{2}}\cdots d_{l_{n}}\sum\limits_{q=2}^{n}\sum\limits_{\alpha_{l_{q}}=1}^{d_{l_{q}}^{2}}
\mu_{\alpha_{l_{2}}\cdots\alpha_{l_{n}}}^{z}A_{\alpha_{l_{2}}}^{(l_{2})}\otimes \cdots\otimes A_{\alpha_{l_{n}}}^{(l_{n})}.
\end{equation}
Let $T^{(l_1)}$ and $T^{(l_{2}\cdots l_{n})}$ be the column vectors with entries of $\mu_{\alpha_{l_{1}}}^{z}$ and $\mu_{\alpha_{l_{2}}\cdots\alpha_{l_{n}}}^{z}$, respectively. We have

\begin{equation}
\begin{split}
\|B^{l_{1}|l_{2}\cdots l_{n}}\|_{tr}&\leq\sum_zr_z\|(T^{(l_1)})(T^{(l_{2}\cdots l_{n})})^{\dagger}\|_{tr}=\sum_zr_z\|T^{(l_1)}\|\|T^{(l_{2}\cdots l_{n})}\|\\
&\leq\sum_zr_z\sqrt{\frac{1}{d_{l_{1}}}}\sqrt{\frac{1}{d_{l_{2}}\cdots d_{l_{n}}}}=\sqrt{\frac{1}{d_{l_{1}}d_{l_{2}}\cdots d_{l_{n}}}},\\
\end{split}
\end{equation}
where we have used $\||a\rangle\langle b|\|_{tr}=\||a\rangle\|\||b\rangle\|$ for vectors $|a\rangle$ and $|b\rangle$, Lemma 1 and Lemma 3 in the second inequality.\\
(ii) If $n$-partite quantum state $\rho$ is separable under the partition $l_{1}\cdots l_{v_{1}}|l_{v_{1}+1}\cdots l_{v_{1}+v_{2}}|\cdots\\|l_{n-v_{k}+1}\cdots l_{n}$, then $\rho$ can be expressed as
\begin{equation}
\rho=\sum_{z}r_{z}\rho_{l_{1}\cdots l_{v_{1}}}^{z}\otimes\rho_{l_{v_{1}+1}\cdots l_{v_{1}+v_{2}}}^{z}\otimes\cdots\otimes\rho_{l_{n-v_{k}+1}\cdots l_{n}}^{z},
\end{equation}
where
\begin{equation}
\rho_{l_{1}\cdots l_{v_{1}}}^{z}=d_{l_1}\cdots d_{l_{v_{1}}}\sum_{q=1}^{v_1}\sum_{\alpha_{l_q}=1}^{d_{l_q}^{2}}\mu_{\alpha_{l_{1}}\cdots\alpha_{l_{v_{1}}}}^{z}A_{\alpha_{l_{1}}}^{(l_1)}\otimes\cdots\otimes A_{\alpha_{l_{v_1}}}^{(l_{v_1})},
\end{equation}
\begin{equation}
\rho_{l_{v_1+1}\cdots l_{v_{1}+v_2}}^{z}=d_{l_{v_1+1}}\cdots d_{l_{v_{1}+v_{2}}}\sum_{q=v_1+1}^{v_1+v_2}\sum_{\alpha_{l_q}=1}^{d_{l_q}^{2}}\mu_{\alpha_{l_{v_1+1}}\cdots\alpha_{l_{v_{1}+v_2}}}^{z}A_{\alpha_{l_{v_1+1}}}^{(l_{v_{1}+1})}\otimes\cdots\otimes A_{\alpha_{l_{v_1+v_2}}}^{(l_{v_1+v_2})},
\end{equation}
$\cdots$
\begin{equation}
\rho_{l_{n-v_{k}+1}\cdots l_{n}}^{z}=d_{l_{n-v_{k}+1}}\cdots d_{l_{n}}\sum_{q=n-v_{k}+1}^{n}\sum_{\alpha_{l_q}=1}^{d_{l_q}^{2}}\mu_{\alpha_{l_{n-v_{k}+1}}\cdots\alpha_{l_{n}}}^{z}
A_{\alpha_{l_{n-v_{k}+1}}}^{(l_{n-v_{k}+1})}\otimes\cdots\otimes A_{\alpha_{l_{n}}}^{(l_{n})}.
\end{equation}
Let $T^{(l_{1}\cdots l_{v_{1}})}$, $T^{(l_{v_1+1}\cdots l_{v_{1}+v_2})}$, $\ldots$, $T^{(l_{n-v_{k}+1}\cdots l_{n})}$ be the column vectors with entries of $\mu_{\alpha_{l_{1}}\cdots\alpha_{l_{v_{1}}}}^{z}$, $\mu_{\alpha_{l_{v_1+1}}\cdots\alpha_{l_{v_{1}+v_2}}}^{z}$, $\ldots$, $\mu_{\alpha_{l_{n-v_{k}+1}}\cdots\alpha_{l_{n}}}^{z}$, respectively. And $T_{j}^{(l_{n-v_{k}+1}\cdots l_{n})}$ be the column vectors composed with a part of $\mu_{\alpha_{l_{n-v_{k}+1}}\cdots\alpha_{l_{n}}}^{z}$, $\alpha_{l_{n}}=(j-1)d_{l_n}+1,\ldots,jd_{l_n}$, $j=1,2,\ldots,d_{l_n}$. We have
\begin{equation}
\begin{split}
B^{l_{1}\cdots l_{v_{1}}|l_{v_{1}+1}\cdots l_{v_{1}+v_{2}}|\cdots|l_{n-v_{k}+1}\cdots l_{n}}=&\sum_{z}r_z(T^{(l_{1}\cdots l_{v_{1}})})^{\dag}\otimes(T^{(l_{v_1+1}\cdots l_{v_{1}+v_2})})^{\dag}\otimes\cdots\otimes(T^{(l_{n-v_{k}-v_{k-1}+1}\cdots l_{n-v_k})})^{\dag}\\
&\otimes\begin{pmatrix}
T_{1}^{(l_{n-v_{k}+1}\cdots l_{n})}~~T_{2}^{(l_{n-v_{k}+1}\cdots l_{n})}~~\cdots~~T_{d_{l_n}}^{(l_{n-v_{k}+1}\cdots l_{n})}\\
\end{pmatrix}
\end{split}
\end{equation}
Then we derive
\begin{equation}
\begin{split}
&~~~~\|B^{l_{1}\cdots l_{v_{1}}|l_{v_{1}+1}\cdots l_{v_{1}+v_{2}}|\cdots|l_{n-v_{k}+1}\cdots l_{n}}\|_{tr}\\
&\leq\sum_{z}r_z\|(T^{(l_{1}\cdots l_{v_{1}})})^{\dag}\|_{tr}\|(T^{(l_{v_1+1}\cdots l_{v_{1}+v_2})})^{\dag}\|_{tr}\cdots\|(T^{(l_{n-v_{k}-v_{k-1}+1}\cdots l_{n-v_k})})^{\dag}\|_{tr}\\
&~~~~~~~\|\begin{pmatrix}
T_{1}^{(l_{n-v_{k}+1}\cdots l_{n})}~~T_{2}^{(l_{n-v_{k}+1}\cdots l_{n})}~~\cdots~~T_{d_{l_n}}^{(l_{n-v_{k}+1}\cdots l_{n})}\\
\end{pmatrix}\|_{tr}\\
&\leq\sum_{z}r_z\|T^{(l_{1}\cdots l_{v_{1}})}\|\|T^{(l_{v_1+1}\cdots l_{v_{1}+v_2})}\|\cdots\|T^{(l_{n-v_{k}-v_{k-1}+1}\cdots l_{n-v_k})}\|\sqrt{d_{l_n}}\|T^{(l_{n-v_{k}+1}\cdots l_{n})}\|\\
&\leq\sum_{z}r_z\sqrt{\frac{1}{d_{l_1}\cdots d_{l_{v_1}}}}\sqrt{\frac{1}{d_{l_{v_1+1}}\cdots d_{l_{v_1+v_2}}}}\cdots\sqrt{\frac{1}{d_{l_{n-v_{k}-v_{k-1}+1}}\cdots d_{l_{n-v_k}}}}\sqrt{d_{l_n}}\sqrt{\frac{1}{d_{l_{n-v_{k}+1}}\cdots d_{l_{n}}}}\\
&=\sqrt{\frac{1}{d_{l_1}\cdots d_{l_{n-1}}}},
\end{split}
\end{equation}
where we have used $\|A\otimes B\|_{tr}=\|A\|_{tr}\|B\|_{tr}$ for matrices $A$ and $B$,
$\|A^{m\times n}\|_{tr}\leq\sqrt{\mathrm{min}\{m,n\}}\\\|A^{m\times n}\|$
 in the second inequality and Lemma 3 in the third inequality.
\end{proof}

In the following example, we use the four matrices from (\ref{33}).

\textbf{Example 3} Consider the four-qubit state $\rho\in H_1^{2}\otimes H_2^{2}\otimes H_3^{2}\otimes H_4^{2}$,
\begin{equation}
\rho=\frac{1-x}{16}\mathbb{I}_{16}+x|\psi\rangle\langle\psi|,~~0\leq x\leq1,
\end{equation}
where $|\psi\rangle=\frac{1}{\sqrt{2}}(|0000\rangle+|1111\rangle)$, $\mathbb{I}_{16}$ is the $16\times16$ identity matrix. From Theorem 4, we obtain that if $f_{2}(x)=\|B^{l_{1}|l_{2}l_{3}l_{4}}\|_{tr}-\frac{1}{4}
=\frac{3|x|}{8}+\frac{\sqrt{3x^{2}+1}}{16}-\frac{1}{4}>0$, i.e., $0.4545<x\leq1$, $\rho$ is entangled under the bipartition $l_{1}|l_{2}l_{3}l_{4}$. By the Theorem 3 in \cite{ref28}, $\rho$ is entangled under the bipartition $l_{1}|l_{2}l_{3}l_{4}$ if $g_{3}(x)=9x^{2}-4>0$, i.e., $0.6667<x\leq1$.

Using Theorem 4, we have obtained that $\rho$ is entangled under the partition $l_1l_2|l_3l_4$ for $f_3(x)=\frac{\sqrt{x^2+1}}{16}+\frac{7\sqrt{2}|x|}{16}-\sqrt{\frac{1}{8}}>0$, i.e., $0.4602<x\leq1$. The Theorem 4 in \cite{ref11} says that $\rho$ is entangled under the partition $l_1l_2|l_3l_4$ for $g_4(x)=\sqrt{1+x^2}+2\sqrt{2}x+\frac{x-x^2}{1+x^2}-4>0$, i.e., $0.915<x\leq1$. The result in \cite{ref29} implies that $\rho$ is entangled (not biseparable) if $\frac{17}{21}$ $(\approx 0.8095)<x\leq1$. Thus, our results outperform these existing results in detecting the entanglement, see Figure 2.
\begin{figure}[h!]
  \centering
  \includegraphics[width=17cm]{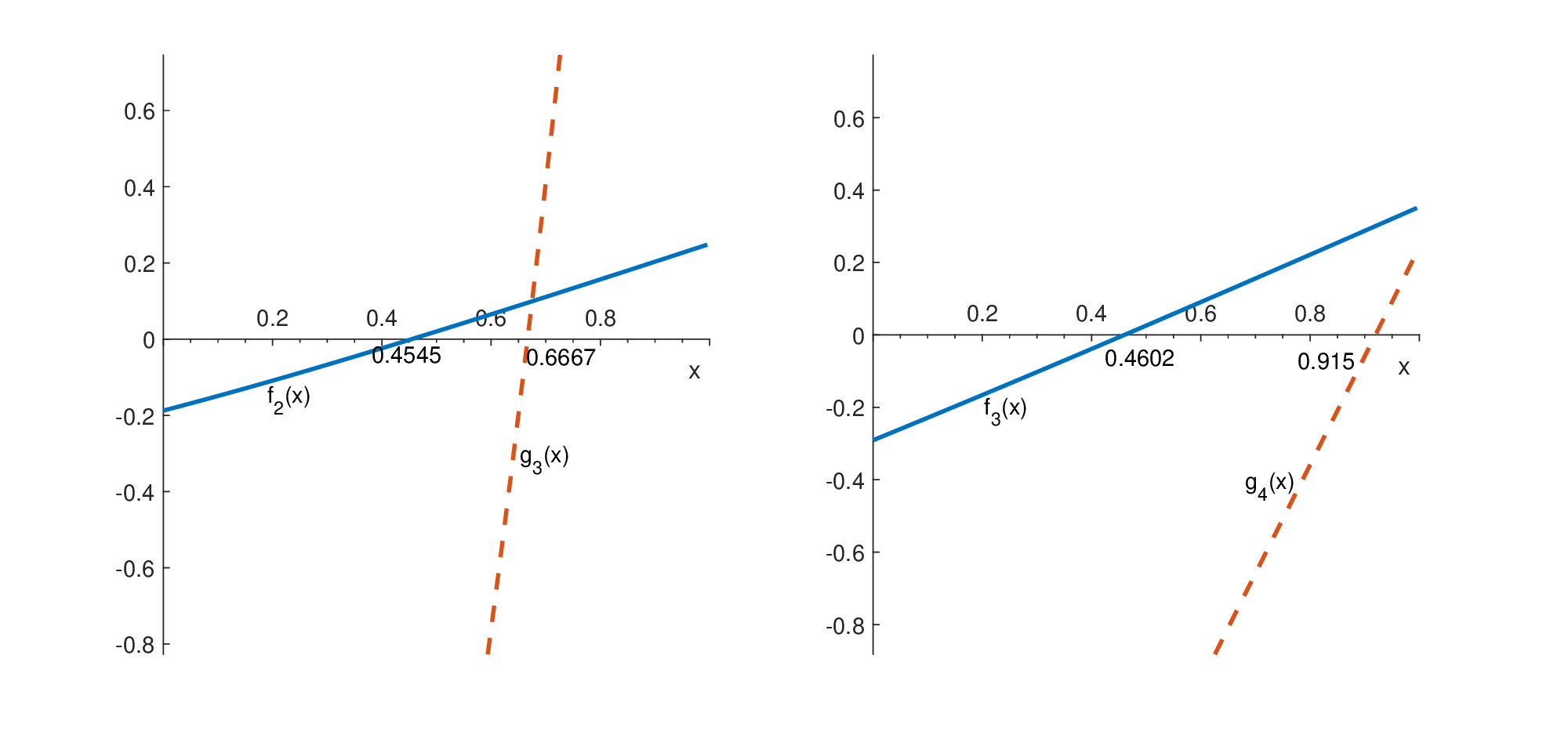}
  \caption{$f_{2}(x)$ from Theorem 4 (solid blue line), $g_{3}(x)$ from Theorem 3 in \cite{ref28} (dotted red line), $f_{3}(x)$ from Theorem 3 (solid blue line), $g_{4}(x)$ from Theorem 4 in \cite{ref11} (dotted red line).}\label{1}
\end{figure}

\textbf{Example 4} Consider the quantum state $\rho\in H_1^{2}\otimes H_2^{2}\otimes H_3^{2}\otimes H_4^{2}$,
\begin{equation}
\rho=\frac{1-x}{16}\mathbb{I}_{16}+x|W\rangle\langle W|,~~0\leq x\leq1,
\end{equation}
where $|W\rangle=\frac{1}{2}(|0001\rangle+|0010\rangle+|0100\rangle+|1000\rangle)$. By using Theorem 4, we obtain $\rho$ is entanglement under the bipartition $l_1|l_2l_3l_4$ for $f_4(x)=B^{l_1|l_2l_3l_4}-\frac{1}{4}>0$, i.e., $0.4891<x\leq1$. Theorem 3 in \cite{ref11} shows that $\rho$ is entanglement under the partition $l_1|l_2l_3l_4$ for $g_5(x)=\frac{4+2x^2}{2\sqrt{4+x^2}}+x-2>0$, i.e., $0.783<x\leq1$. Figure 3 shows that our method can detects more entanglement.
\begin{figure}[h]
  \centering
  \includegraphics[width=15cm]{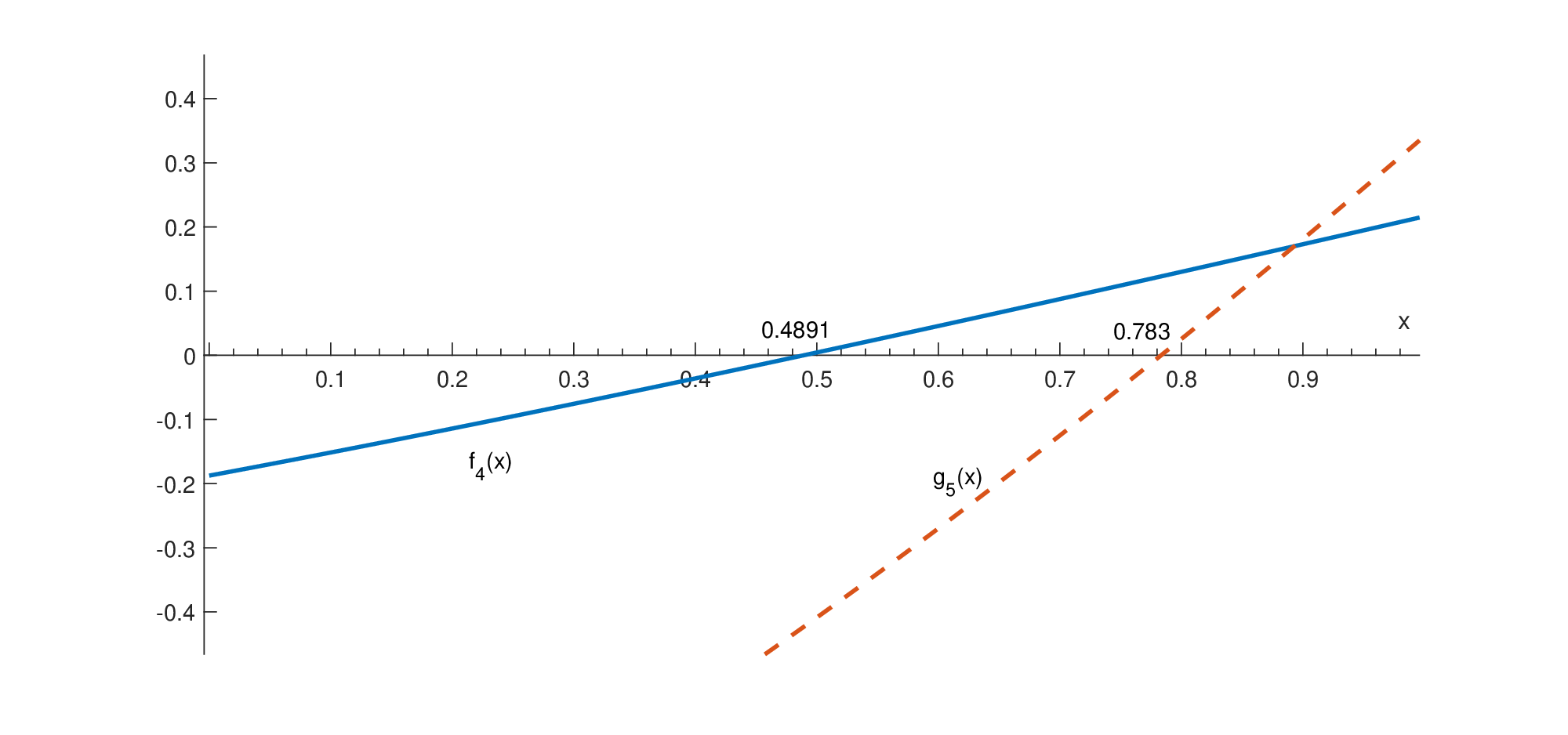}
  \caption{$f_{4}(x)$ from Theorem 3 (solid blue line), $g_{5}(x)$ from Theorem 3 in \cite{ref11} (dotted red line).}\label{1}
\end{figure}

\section{Conclusion}
By using GSICM, COB and their relations, we have derived finer upper bounds for the norms of correlation probabilities and presented the criteria for detecting the genuine tripartite entanglement. And then we have studied multipartite entanglement for arbitrary dimensional multipartite systems. Detailed examples demonstrate that our criteria can detect better genuine entanglement and multipartite entanglement.
The method of detecting entanglement with GSICM depends on some local measurements. While GSICM has been shown to possess operational advantages in many quantum tasks over projective measurements. Based on the relationship between GSICM and COB, the trace relationship of COB depends only on the dimensions. Thus the detection of entanglement in COB could be relatively easy to be implemented experimentally.
Furthermore, the authors in \cite{ref31} show that GSICM can be applied to study different types of quantum correlations. Therefore, our approach may highlight further investigations on detection of other quantum correlations.

\medskip
\noindent\textbf{ CRediT authorship contribution statement}

{\bf Hui Zhao, Jia Hao:} Formal analysis, Writing, Calculation and figure. {\bf Jing Li, Shao-Ming Fei, Naihuan Jing, Zhi-Xi Wang:} Writing-review and editing.\\

\noindent\textbf{Declaration of competing interest}

The authors declare that they have no known competing financial interests or personal relationships that could have appeared to
influence the work reported in this paper.\\

\noindent\textbf{Data availability}

No data was used for the research described in the article.\\

\noindent\textbf {Acknowledgements}
This work is supported by the National Key R\&D Program of China under Grant No.(2022YFB3806000), National Natural Science Foundation of China under Grants (12272011, 12075159, 12126351 and 12171044), Beijing Natural Science Foundation (Z190005), the Academician Innovation Platform of Hainan Province.


\end{document}